\documentclass[11pt]{article}
\usepackage[margin=1in]{geometry}
\usepackage{amsmath,amssymb,amsthm}
\usepackage{booktabs,array,longtable}
\usepackage[hidelinks]{hyperref}
\usepackage{xcolor}

\newtheorem{theorem}{Theorem}[section]
\newtheorem{lemma}[theorem]{Lemma}
\newtheorem{proposition}[theorem]{Proposition}
\newtheorem{corollary}[theorem]{Corollary}

\newcommand{\1}{\mathbf 1}

\title{The Deterministic Hare Core Is Nonempty for Nine-Seat Approval Elections}
\author{Jiarui Fang\\[0.35em]
\small Boston University\\
\small \href{mailto:baymin@bu.edu}{\texttt{baymin@bu.edu}}\\
\small ORCID: \href{https://orcid.org/0009-0006-9100-0445}{0009-0006-9100-0445}}
\date{September 2026}

\begin{document}
\maketitle

\begin{abstract}
Core stability gives approval-based committee elections a strong form of
coalitional proportionality, but deterministic existence is not known for an
arbitrary number of seats.  We prove that every finite election with nine
seats has a deterministic Hare-quota core committee.  Starting from an
eight-seat committee that is both a global Proportional Approval Voting (PAV)
maximizer and core stable, we suppose that all one-candidate extensions are
blocked and extract an inclusion-minimal closed response system.  Singleton
responses are eliminated by exact equality rigidity.  For the remaining
responses, rational Farkas certificates bound every nonpositive PAV
add-marginal drift below by $-n/189$ and force every successor drift above
$n/126$.  The uniform response chain is irreducible, so stationarity makes
these bounds incompatible.  The computer-assisted component comprises an
exact pointwise check over 7,356 voter types, 36 one-blocker cells, and a
symmetry-complete family of 19 two-blocker motifs.  Exact certificate archives
and standard-library verifiers support these checks.  The theorem
settles the nine-seat case but does not decide deterministic core existence
for arbitrary committee size.
\end{abstract}

\section{Introduction}

Approval-based committee elections select a fixed number of candidates from
binary approval ballots.  A committee is core stable when no coalition can
use its proportional share of the seats to make every coalition member
strictly better off.  This requirement simultaneously captures individual
utility and group entitlement, and it is therefore stronger than many
candidate-by-candidate proportionality axioms \cite{Peters2025,Cheng2019}.

Whether a deterministic Hare-quota core committee always exists remains a
central unresolved question.  Peters proved existence for at most eight seats
and, independently, for at most fifteen candidates \cite{Peters2025}.
Berker et al. developed an exact blocker-map duality and further special-case
results, while retaining unrestricted existence as open
\cite{Berker2026}.  Becker, Greger, and Peters proved existence for profiles
with at most seven voters and, in the weighted formulation, at most seven
distinct approval types \cite{Becker2026}.  Their affine-monoid and Lindahl
rounding arguments restrict voter types rather than committee size.  Chen and
Hatschka embedded the problem in transferable-utility (TU) and
non-transferable-utility (NTU) voting games \cite{ChenHatschka2026}.  Their
NTU cores for AV, SAV, and PAV represent the same committee-stability object
studied here, whose unrestricted existence they leave open.  Their positive TU
and NTU-CC results use different utility-transfer or scoring assumptions.
Stable lotteries exist much more generally \cite{Cheng2019}, but their
quantifiers do not yield one deterministic committee that defeats every
integral deviation.  Relative to Peters's unrestricted voter-and-candidate
bound, the present theorem extends the verified committee-size range from
eight to nine.  The release includes a dated, bounded source audit and makes no
claim of an exhaustive priority search.

The main obstruction is global rather than local.  If a target blocks an
extension $U\cup\{c\}$ of an eight-seat core $U$, replacing $c$ by a new
outsider from that target escapes the same blocker.  The repaired extension
may, however, acquire another blocker.  An argument about one blocker at a
time cannot exclude a closed cycle of such responses.  We therefore model all
blockers in an inclusion-minimal closed system inside one voter profile.

The proof has three parts.  First, exact equality analysis rules out singleton
responses in a closed system.  Second, one-blocker and simultaneous
two-blocker Farkas certificates convert the remaining blocking constraints
into exact certified PAV drift bounds.  Third, the response sets define an irreducible
Markov chain.  Its stationary distribution telescopes the PAV marginals, but
the certified drift bounds force a strictly positive stationary mean.  This
contradiction rules out every closed blocker system.  All symbolic reductions,
certificate formats, exact verifiers, and the trusted-computing-base boundary
are included in the release package.

\section{Model and theorem}

Let $N$ be a nonempty finite voter set, $C$ a finite candidate set, and
$A_i\subseteq C$ the approval set of voter $i$.  Put $n=|N|$ and
\[
        u_i(X)=|A_i\cap X| \qquad (X\subseteq C).
\]
At committee size $k$, a committee is a $k$-set $W\subseteq C$.  For a
nonempty coalition $S\subseteq N$, a nonempty target $T\subseteq C$ is
\emph{affordable} if
\[
                         n|T|\le k|S|.
\]
The coalition $S$ \emph{blocks} $W$ through $T$ if $T$ is affordable and
$u_i(T)>u_i(W)$ for every $i\in S$.  We say that $T$ blocks $W$ if some
nonempty coalition blocks through it.  The target may overlap $W$, and
equality in the displayed affordability condition is allowed.  For a
nonempty target $T\subseteq C$ define
\[
 q(W,T)=|\{i\in N:u_i(T)>u_i(W)\}|.
\]

\begin{proposition}[integer form of blocking]\label{prop:integer}
A nonempty target $T$ blocks $W$ if and only if
\[
             kq(W,T)\ge n|T|.                              \tag{1}
\]
Consequently $W$ is in the core if and only if
\[
             kq(W,T)<n|T|                                  \tag{2}
\]
for every nonempty $T\subseteq C$ with $|T|\le k$.
\end{proposition}

\begin{proof}
If a coalition $S$ blocks through $T$, then every member of $S$ is counted
by $q(W,T)$, so $q(W,T)\ge |S|$; affordability gives
$k|S|\ge n|T|$, and (1) follows.  Conversely, if (1) holds, take $S$ to be
the entire strict-gainer set.  It is nonempty because $T$ is nonempty and
$n|T|>0$.  Then $k|S|\ge n|T|$, so $T$ is affordable and every member of
$S$ strictly improves.  Finally $|S|\le n$ implies $|T|\le k$.  This proves
both assertions without rounding or a ceiling convention.
\end{proof}

Our main result is the following exact deterministic existence theorem.

\begin{theorem}[nine-seat nonemptiness]\label{thm:main}
For every finite approval election with $|C|\ge9$, there is a set
$W\subseteq C$, $|W|=9$, satisfying (2) for every nonempty target of size at
most nine.
\end{theorem}

\section{The certified eight-seat anchor}

Let $H_s=\sum_{j=1}^s1/j$, with $H_0=0$, and define the PAV score
\[
             \Phi(X)=\sum_{i\in N}H_{u_i(X)}.
\]
We use the following established exact theorem.

\begin{theorem}[Peters's eight-seat theorem]\label{thm:k8}
For every finite approval election, at least one global maximizer of
$\Phi$ among the eight-element committees is an eight-seat Hare-core
committee.
\end{theorem}

Theorem~\ref{thm:k8} is Theorem 4.5 of Peters
\cite{Peters2025}; its source uses rational certificates rather than
floating-point feasibility tolerances.  If some voters have empty
approval sets, delete them before invoking the theorem.  Reinserting them
does not change any strict-gainer count or PAV score and only increases the
right-hand side $n|T|$ in (2).  If every ballot is empty, every committee is
core.  Thus Theorem~\ref{thm:k8} applies to the model here, including repeated
types and empty ballots.

Fix an eight-set $U$ supplied by Theorem~\ref{thm:k8}.  Write
\[
 O=C\setminus U,\qquad W_c=U\cup\{c\}\quad(c\in O),
 \qquad
 a_c=\sum_{i:c\in A_i}\frac1{u_i(U)+1}.                    \tag{3}
\]
Thus $a_c=\Phi(U\cup\{c\})-\Phi(U)$ is the exact PAV add marginal.

\section{Blockers as a closed response system}

\begin{lemma}[cancellation and same-target escape]\label{lem:cancel}
If $T$ blocks $W_c$, then $c\notin T$ and
$B(T):=T\setminus U$ is nonempty.  Moreover, for every $x\in B(T)$, the same
target $T$ does not block $W_x$.
\end{lemma}

\begin{proof}
Suppose first that $c\in T$, put $t=|T|$, and cancel $c$ from both sides of
each strict utility comparison.  The same $q=q(W_c,T)$ voters strictly
prefer $T\setminus\{c\}$ to $U$.  The reduced target is nonempty, since
$T=\{c\}\subseteq W_c$ has no strict gainer.  By the literal integer quota,
\[
 8q=9q-q\ge nt-q\ge n(t-1),
\]
so $T\setminus\{c\}$ would block the eight-seat core $U$, a contradiction.
If $T\setminus U$ were empty, then $T\subseteq U\subseteq W_c$ and again
there would be no strict gainer.  Finally, if the same $T$ blocked $W_x$ for
$x\in T\setminus U$, cancelling $x$ would give the identical eight-seat
contradiction.  Equality in the quota is included throughout.
\end{proof}

A nonempty set $X\subseteq O$ is a \emph{closed fixed-$U$ blocker system} if
for every $c\in X$ there is a blocker $T_c$ of $W_c$ whose response set
\[
             \varnothing\ne B_c:=T_c\setminus U
             \subseteq X\setminus\{c\}.                    \tag{4}
\]
If every extension $W_c$ is blocked, $O$ is such a system by
Lemma~\ref{lem:cancel}.  Finiteness then supplies an inclusion-minimal one.

\begin{lemma}[minimal trap is irreducible]\label{lem:strong}
Let $X$ be an inclusion-minimal closed system and choose any one internal
blocker satisfying (4) at every $c\in X$.  The digraph containing all arcs
$c\to x$ for $x\in B_c$ is strongly connected.
\end{lemma}

\begin{proof}
If the condensation had more than one vertex, a sink strongly connected
component $Y$ would be a proper nonempty subset of $X$.  Every chosen
response from a state in $Y$ remains in $Y$, so the same chosen blockers
would certify that $Y$ is closed, contrary to minimality.
\end{proof}

\section{Singleton responses cannot occur in a closed system}

This section records the equality rigidity needed to ensure that every
selected response set has at least two outsiders.  All statements concern
the same fixed global PAV maximizer $U$.

\begin{lemma}[singleton rigidity]\label{lem:singleton}
Suppose $H\subseteq U$, $x,c\notin U$, $x\ne c$, and
$T=H\cup\{x\}$ blocks $W_c$.  Put $h=|H|$, $D=U\setminus H$, and let $Q$
be the strict-gainer set.

If $h\le7$, then
\begin{enumerate}
\item $9|Q|=n(h+1)$;
\item voters in $Q$ approve $x$, disapprove $c$, and approve exactly $H$
      within $U$;
\item voters outside $Q$ disapprove $x$ and have a nonempty $U$-approval
      set contained in $D$; and
\item every swap $U-u+x$, $u\in U$, ties $U$ in PAV score.
\end{enumerate}
If $h=8$, every voter approves $x$ and disapproves $c$, every member of $U$
is universally approved, and the same eight swaps tie.
\end{lemma}

\begin{proof}
Assume $h\le7$ and put $d=8-h$.  For voter $i$, set
$s_i=u_i(U)$ and $z_i=|A_i\cap D|$.  Sum, over $u\in D$, that voter's PAV
changes in $U-u+x$.  The contribution is
\[
 \frac{d-z_i}{s_i+1}\quad(x\in A_i),\qquad
 -\frac{z_i}{s_i}\quad(x\notin A_i, s_i>0),
 \qquad 0\quad(x\notin A_i, s_i=0).                       \tag{5}
\]
The total $S_x$ is nonpositive by global PAV maximality.  A gainer in $Q$
approves $x$, disapproves $c$, has $z_i=0$ and $s_i\le h$, so contributes at
least $d/(h+1)$.  Every other voter contributes at least $-1$.  Hence
\[
 0\ge S_x\ge |Q|\frac d{h+1}-(n-|Q|)
       =\frac{9|Q|}{h+1}-n\ge0,                             \tag{6}
\]
where the last inequality is exactly $9|Q|\ge n(h+1)$.
Equality at every step yields the first three assertions.  The sum of the
nonpositive labelled swap changes is zero, so all swaps indexed by $D$ tie;
the derived incidence makes swaps indexed by $H$ tie as well.

If $h=8$, the target has size nine, so its block forces $|Q|=n$.
Strict comparison of $U+x$ with $U+c$ says that every voter approves $x$
and disapproves $c$.  A voter disapproving some $u\in U$ would then strictly
gain in $U-u+x$, while nobody could lose, contradicting global maximality.
The final assertions follow.
\end{proof}

\begin{lemma}[tied-candidate absorption]\label{lem:absorb}
Let $x\notin U$ and suppose $\Phi(U-u+x)=\Phi(U)$ for every $u\in U$.
For a target $T=H\cup R$ excluding $x$, where $H\subseteq U$ and
$\varnothing\ne R\subseteq C\setminus(U\cup\{x\})$, put $h=|H|$ and
$r=|R|$, and suppose $h+r\le9$.  Then
\[
 \frac{q(W_x,T)}n\le
 \begin{cases}
 (h+r-1)/9,&h\ge1,\\
 r/10,&h=0, r\ge2,\\
 1/9,&h=0, r=1.
 \end{cases}                                               \tag{7}
\]
Thus only a bare outsider singleton can block $W_x$.
\end{lemma}

\begin{proof}[Exact pointwise certificate proof]
Average ballots under permutations inside $H$, $U\setminus H$, and $R$.
A voter type is given by the four integers
\[
 p=|A_i\cap H|,\quad s=|A_i\cap(U\setminus H)|,
 \quad z=\1[x\in A_i],\quad y=|A_i\cap R|,                 \tag{8}
\]
and is a strict gainer precisely when $y>s+z$.  For $1\le h\le7$, put
$d=8-h$ and let $E_H,E_D$ be the type's average PAV changes in $U-u+x$
over $u\in H,D$, respectively.  Let $\Delta_B$ be the average change in
$U-u+v$ over $u\in H,v\in R$, and $\Delta_A$ the average change in
$U-\{u,w\}+\{x,v\}$ over $u\in H,w\in D,v\in R$.  Direct expansion of
harmonic increments gives, for every integral type (8),
\[
 \1[y>s+z]\le b+\alpha\Delta_A+\beta\Delta_B
                 +\gamma E_H+\eta E_D,                     \tag{9}
\]
where
\[
 b=\frac{h+r-1}{9},\quad \alpha=\frac{rd}{d+2},\quad
 \beta=\frac r{d+2},\quad
 \gamma=-\frac{(d+1)(h+r-1)}9,
 \quad \eta=\frac{d(h-1)(d+2-r)}{9(d+2)}.                 \tag{10}
\]
The aggregate $E_H,E_D$ terms vanish by the eight ties, while the aggregate
$\Delta_A,\Delta_B$ terms are nonpositive by global PAV optimality and
$\alpha,\beta\ge0$.  Summing (9) proves the first line of (7).

For $h=0,r\ge2$, let $E$ average $U-u+x$, let $\Delta_B$ average
$U-u+v$, and let $\Delta_A$ average
$U-\{u,w\}+\{x,v\}$.  The exact pointwise certificate is
\[
 \1[y>s+z]\le \frac r{10}+\frac{28r}{45}\Delta_A
       +\frac{8r}{45}\Delta_B-\frac{32r}{45}E.            \tag{11}
\]
For $h=8$ the corresponding certificate is
\[
 \1[y>z]\le \frac89+\frac19\Delta_B-\frac89E_H.          \tag{12}
\]
The same summation proves the appropriate bounds.  Finally, when
$h=0,r=1$, a gainer approves the unique member $v$ of $R$ and has
$u_i(U)=z=0$.  Summing the eight inequalities
$\Phi(U-u+v)-\Phi(U)\le0$ gives $0\ge9q-n$, proving the last line.

For auditability, the standard-library verifier included with this package
reconstructs every harmonic average in (9)--(12) and checks all $6,384$
middle types, $936$ empty-$H$ types, and $36$ $h=8$ types using
\texttt{fractions.Fraction}.  Thus ``direct expansion'' here is a finite
exact identity check, not a floating-point assertion.
\end{proof}

\begin{lemma}[no two singleton responses in succession]\label{lem:noss}
After a singleton block $H\cup\{x\}$ of $W_c$, no target with exactly one
outsider can block $W_x$.
\end{lemma}

\begin{proof}
For the first block with $h=8$, every voter approves $x$, whereas a singleton
response blocker of $W_x$ needs a positive mass of gainers disapproving $x$.
Suppose $h\le7$ and a second singleton target $H'\cup\{y\}$ blocks $W_x$.
Apply Lemma~\ref{lem:singleton} to both blocks.  The second gainer set $Q'$
disapproves $x$ and has $U$-approval set exactly $H'$.  By the first block's
rigidity, $H'$ is nonempty and $H'\subseteq D=U\setminus H$.  Pick
$u\in H'$.  In the tied swap $U-u+x$, the first gainer set contributes
\[
        \frac{|Q|}{h+1}=\frac n9,
\]
while the second gainer set contributes
\[
        -\frac{|Q'|}{h'}=-\frac{n(h'+1)}{9h'}<-\frac n9.
\]
Every remaining voter either loses or is unchanged, by the first rigidity
partition.  The swap is therefore strictly PAV-worse than $U$, contradicting
that it ties.  The $h'=0$ and $h'=8$ boundary cases contradict, respectively,
the nonempty-$U$ support of first-block nongainers and the positive mass of
first-block gainers who approve $x$.
\end{proof}

\begin{corollary}[response size]\label{cor:r2}
Every internal blocker in a closed fixed-$U$ system has
$|B_c|\ge2$.
\end{corollary}

\begin{proof}
If an internal blocker had $B_c=\{x\}$, Lemma~\ref{lem:singleton} would make
all $x$-swaps tie.  Because $x$ lies in the closed system, $W_x$ has an
internal blocker.  Lemma~\ref{lem:absorb} says that it can only be a bare
singleton, while Lemma~\ref{lem:noss} forbids it.
\end{proof}

\section{One-blocker PAV drift}

For a selected blocker write
\[
 T_c=H_c\cup B_c,\qquad H_c=T_c\cap U,\qquad
 h_c=|H_c|,\quad r_c=|B_c|,
\]
and define its add-marginal drift
\[
 \delta_c=\frac1{r_c}\sum_{x\in B_c}a_x-a_c.              \tag{13}
\]

\begin{lemma}[exact 36-cell drift theorem]\label{lem:drift}
If $T_c$ blocks $W_c$ and $r_c\ge2$, then
\[
                 \delta_c\ge-\frac n{189}.                 \tag{14}
\]
If $h_c=0$ or $r_c\ge4$, then
\[
                 \delta_c\ge\frac n{126}.                  \tag{15}
\]
In particular, $\delta_c\le0$ implies $h_c>0$ and
$r_c\in\{2,3\}$.
\end{lemma}

The stronger exact cell table certified in the package is as follows; each
entry is a lower bound on $\delta_c/n$.
\[
\begin{array}{c|rrrrrrrr}
 r\backslash h&0&1&2&3&4&5&6&7\\ \hline
2&1/72&0&-1/189&-1/216&0&1/108&1/18&1/9\\
3&1/72&1/324&-1/567&0&1/135&1/18&1/9&\\
4&1/72&53/6048&1/126&17/1080&1/18&1/9&&\\
5&38/2259&119/6012&29893/1097460&1/18&1/9&&&\\
6&2005/93618&445/14013&1/18&1/9&&&&\\
7&170/5403&1/18&1/9&&&&&\\
8&1/18&1/9&&&&&&\\
9&1/9&&&&&&&
\end{array}                                                \tag{16}
\]

\begin{proof}[Certificate proof of Lemma~\ref{lem:drift}]
Normalize voter mass to one and average over permutations inside
$H_c$, $U\setminus H_c$, and $B_c$, fixing $c$.  A voter orbit is
$(y,z,e,b)$, the numbers approved in these first three classes and the bit
for $c$.  Put $s=y+z$.  Its exact strict-gainer indicator and drift
contribution are
\[
 g=\1[e>z+b],\qquad d=\frac{e/r-b}{s+1}.                   \tag{17}
\]
The rational system retains the sum of all $8r$ one-swap PAV rows, the
matching weak eight-seat core row, every orbit-average of an integral
eight-seat PAV competitor in $U\cup B_c\cup\{c\}$, and the exact block row
\[
             -\sum_v g(v)p_v\le-\frac{h+r}{9}.             \tag{18}
\]
All are necessary consequences of the hypotheses; for a target of size
nine the weak core row is merely the valid bound $q/n\le1<9/8$.

Write these rows as $R_jp\le b_j$.  For each of the 36 pairs
$h\ge0,r\ge2,h+r\le9$, the JSON archive gives nonnegative rational
$\lambda_j$ and a rational $\mu$ such that, for every orbit $v$,
\[
 d(v)+\sum_j\lambda_jR_j(v)\ge\mu,\qquad
 L(h,r)=\mu-\sum_j\lambda_jb_j.                            \tag{19}
\]
Multiplying by $p_v$, summing, and using the retained rows proves
$\delta_c/n\ge L(h,r)$.  The independent verifier regenerates all $6,732$
columns and checks (19) in exact rational arithmetic.  Its output verifies
that the only negative cells are $(2,2),(3,2),(2,3)$, the only zero cells
are $(1,2),(4,2),(3,3)$, the minimum for $r\le3$ is $-1/189$, and the
minimum for $r\ge4$ is $1/126$.  The $h=0$ entries in (16) are all at least
$1/72$.  These facts give (14)--(15).
\end{proof}

\section{The simultaneous two-blocker separation}

The next lemma is the key closed-system ingredient.  It models two blockers
in the \emph{same} voter profile; combining two separately feasible local
LPs would not suffice.

\begin{lemma}[destination drift]\label{lem:destination}
Let $T_c=H_c\cup B_c$ block $W_c$, let $x\in B_c$, and let
$T_x=H_x\cup B_x$ block $W_x$.  Suppose
\[
 h_c,h_x>0,\qquad r_c,r_x\in\{2,3\},\qquad \delta_c\le0.  \tag{20}
\]
Then
\[
                         \delta_x>\frac n{126}.             \tag{21}
\]
\end{lemma}

\begin{proof}[Symmetry-complete exact Farkas proof]
Label the edge endpoints $c,x$.  If $r_c=2$, write
$B_c=\{x,a\}$.  The target response $B_x$ is determined, up to relabeling,
by $r_x\in\{2,3\}$, the bit $\1[c\in B_x]$, the bit
$\1[a\in B_x]$, and the number of remaining fresh candidates.  These are
exactly eight motifs.  For every motif, and for every ordered Venn pattern
of $H_c,H_x$ in $U$, the exact archive proves something stronger: the two
blocks cannot coexist at all when $\delta_c\le0$.

If $r_c=3$, write $B_c=\{x,a_1,a_2\}$.  Up to exchange of $a_1,a_2$, the
target is determined by
\[
 r_x,\quad \epsilon=\1[c\in B_x],\quad
 s=|B_x\cap\{a_1,a_2\}|,\quad
 f=r_x-\epsilon-s.                                        \tag{22}
\]
There are five admissible motifs for $r_x=2$ and six for $r_x=3$.  By
Lemma~\ref{lem:drift}, the source condition in (20) restricts
$h_c$ to $2$ or $3$.  The archive covers all 427 resulting Venn orbits and
proves (21).

Here is the certificate common to both cases.  The ordered pair
$(H_c,H_x)$ partitions $U$ into four Venn cells of sizes
$(\nu_{00},\nu_{10},\nu_{01},\nu_{11})$ summing to eight.  Average an
alleged profile under permutations inside these four cells.  In the
source-$r=3$ implementation we also average inside the outsider classes in
(22); its voter orbit records an approval count in each nonempty class and
the two bits for $c,x$.  The source-$r=2$ implementation retains a separate
bit for every outsider role, which is a finer (and therefore still valid)
type space.  Both targets, both strict-gainer predicates, and both drifts are
preserved by the respective reductions.

For each orbit, let $p\ge0$ be the vector of normalized voter masses and
introduce a common block scale $z\ge0$.  The retained exact inequalities are:
\begin{enumerate}
\item every cell-averaged one-swap PAV inequality $\le0$;
\item the two matching weak lower-core inequalities
      $q(U,T_j)/n\le |T_j|/8$;
\item the source row $\delta_c/n\le0$;
\item for the $r_c=3$ package, the contradictory cap
      $\delta_x/n\le1/126$; and
\item the two exact scaled blocking rows
\[
       -9q(W_j,T_j)/n+|T_j|z\le0,\qquad j\in\{c,x\}.       \tag{23}
\]
\end{enumerate}
Every genuine pair satisfying the contrary assumptions is feasible with
$z=1$.  Extra candidates can be deleted from ballots because none occurs in
a retained committee or target.

Write this system as $Ap+\alpha z\le b$, $\mathbf1^Tp=1$.
For every Venn orbit the stored certificate gives $\lambda\ge0$ and rational
$\rho$ with
\[
 \lambda^T\alpha\ge1,\qquad
 \lambda^TA_{\cdot v}+\rho\ge0\ \ (\text{every voter orbit }v),
 \qquad \beta=\lambda^Tb+\rho<1.                           \tag{24}
\]
Multiplication and normalization give $z\le\beta<1$, contradicting
$z=1$.

The source-$r=2$ verifier independently reconstructs all eight motifs and
$1,008$ Venn orbits.  In motif order
$(r_x,\1[c\in B_x],\1[a\in B_x])$, its eight largest exact $\beta$ values
are
\[
 \frac{54}{55},\frac{12}{13},\frac{18}{19},\frac{12}{13},
 \frac{45}{47},\frac{45}{47},\frac{18}{19},\frac9{10}.      \tag{25}
\]
The source-$r=3$ verifier reconstructs all eleven motifs, all $427$ Venn
orbits, and $570,400$ voter-orbit columns; its eleven largest exact $\beta$
values are
\[
 \frac{1464}{1505},\frac{61}{64},\frac{55}{58},\frac{18}{19},
 \frac{387}{442},\frac{143}{147},\frac{143}{147},\frac{12}{13},
 \frac{18}{19},\frac{675}{728},\frac{12}{13},              \tag{26}
\]
all strictly below one.  Together the two packages exhaust exactly the 19
motifs described above, including mutual edges, back edges, shared response
candidates, and wholly fresh responses.  A third, builder-independent checker
explicitly enumerates the response-set permutation action, verifies that
equal signatures are exactly equal orbits, and checks the $1,008+427$ Venn
cell count.  This proves the lemma.
\end{proof}

\section{Stationary-flow contradiction}

The last step has a form that is reusable beyond nine seats.

\begin{lemma}[portable stationary escape criterion]\label{lem:portable}
Let $X$ be finite, let $P$ be an irreducible stochastic matrix on $X$, and
let $a:X\to\mathbb R$.  Put $\delta=Pa-a$.  Suppose there are constants
$0<L<M$ such that $\delta_c\ge-L$ for every $c$, and every transition of
positive probability out of a state with $\delta_c\le0$ lands at a state
with $\delta_x\ge M$.  Then these assumptions are inconsistent.
\end{lemma}

\begin{proof}
Let $\pi>0$ be stationary and put
$F=\{c:\delta_c\le0\}$ and $G=\{c:\delta_c\ge M\}$.  Since
$\sum_c\pi_c\delta_c=\pi Pa-\pi a=0$, the set $F$ is nonempty.  The support
assumption and stationarity give $\pi(G)\ge\pi(F)$.  All states outside
$F\cup G$ have positive drift, so
\[
 0=\sum_c\pi_c\delta_c\ge-L\pi(F)+M\pi(G)
   \ge(M-L)\pi(F)>0,
\]
a contradiction.
\end{proof}

Thus an extension argument at another seat number can reuse the global
closure step verbatim: it only needs a one-blocker loss bound $L$ and a
simultaneous edge certificate with destination gain $M>L$.

\begin{proof}[Proof of Theorem~\ref{thm:main}]
Assume for contradiction that every extension $W_c$, $c\in O$, is blocked.
Choose an inclusion-minimal closed system $X$ and one internal blocker
$T_c$ at each state.  By Corollary~\ref{cor:r2}, $r_c\ge2$ for all $c$.
By Lemma~\ref{lem:strong}, the selected response digraph is strongly
connected.

Define the stochastic matrix
\[
 P(c,x)=\begin{cases}1/r_c,&x\in B_c,\\0,&x\notin B_c.\end{cases}            \tag{27}
\]
It is irreducible, hence has a stationary distribution $\pi$ with
$\pi_c>0$ for every $c\in X$.  Stationarity and (13) telescope exactly:
\[
 \sum_{c\in X}\pi_c\delta_c
 =\sum_c\pi_c\left(\sum_xP(c,x)a_x-a_c\right)=0.           \tag{28}
\]

Put
\[
 F=\{c:\delta_c\le0\},\qquad
 G=\{c:\delta_c\ge n/126\}.                               \tag{29}
\]
The set $F$ is nonempty, or every term in (28) is strictly positive.
Lemma~\ref{lem:drift} gives $\delta_c\ge-n/189$ on $F$ and says that every
$c\in F$ has $h_c>0$ and $r_c\in\{2,3\}$.

Fix $c\in F$ and $x\in B_c$.  If $h_x=0$ or $r_x\ge4$,
Lemma~\ref{lem:drift} puts $x\in G$.  Otherwise
$h_x>0$ and $r_x\in\{2,3\}$, so Lemma~\ref{lem:destination} again puts
$x\in G$.  Therefore
\[
                  P(c,G)=1\quad(c\in F).                   \tag{30}
\]
Stationarity gives the cut inequality
\[
 \pi(G)=\sum_c\pi_cP(c,G)\ge\sum_{c\in F}\pi_c=\pi(F).    \tag{31}
\]
States outside $F\cup G$ have positive drift.  Combining
(14), (28), and (31) now gives
\[
 0=\sum_c\pi_c\delta_c
 \ge-\frac n{189}\pi(F)+\frac n{126}\pi(G)
 \ge\frac n{378}\pi(F)>0,                                 \tag{32}
\]
the desired contradiction.  Hence some $W_c$ is a nine-seat core.
\end{proof}

Notice that the proof uses every target size from one through nine, permits
targets to overlap the committee, and treats quota equality as blocking.
No fractional committee, randomized outcome, Droop quota, or ceiling
convention occurs.

\section{Computer-assisted proof boundary and validation}

The mathematical trusted base consists of Theorem~\ref{thm:k8}, the symbolic
portions of Sections 2--8, and the finite exact verification used in
Lemmas~\ref{lem:absorb}, \ref{lem:drift}, and \ref{lem:destination}.  It also
includes Python's arbitrary-precision integers and \texttt{fractions.Fraction},
three frozen certificate data sets, and five row-reconstructing or coverage
verifiers.  SciPy, NumPy, HiGHS, and all
discovery or certificate-building scripts are outside the trusted base.  The
included exact verifiers use no floating-point comparisons.  They abort if
Python disables assertions.  This package validates the three new frozen
certificate archives.  It does not regenerate them, include solver logs, or
replay the external certificate underlying Theorem~\ref{thm:k8}.

Run these commands from the package root:
\begin{verbatim}
python -B verification/verify_motif_coverage.py
python -B verification/verify_tied_absorption.py
python -B verification/verify_general_blocker_drift.py
python -B verification/verify_r2_edges.py
python -B verification/verify_r3_edges.py
\end{verbatim}
The script \texttt{reproduce.ps1} runs these checks, performs a two-pass
compilation in a separate directory, and verifies the release manifest.  The
build directory is excluded from the manifest.  The script does not assert
byte identity between the rebuilt PDF and the released \texttt{paper/main.pdf}.
The expected leading markers are
\begin{verbatim}
DESTINATION_MOTIF_COVERAGE_EXACT_PASS
TIED_CANDIDATE_ABSORPTION_EXACT_PASS
GENERAL_BLOCKER_DRIFT_FARKAS_EXACT_PASS
DIRECTED_R2_DRIFT_FARKAS_EXACT_PASS
LOW_EDGE_R3_DESTINATION_DRIFT_FARKAS_EXACT_PASS
\end{verbatim}
The package was assembled with Python 3.10.16.  A SHA-256 manifest freezes
every stable release file.  Each verifier reconstructs the relevant voter
orbits, strict-gainer predicates, PAV inequalities, and Farkas signs from the
mathematical definitions before checking the stored rational multipliers.  A
successful run therefore establishes exact certificate validation and source
compilability, not end-to-end certificate generation.

\subsection{Two worked certificate records}

For a one-blocker example, take the archived cell $(h,r)=(0,2)$ and write
$B_c=\{x_1,x_2\}$.  Averaging under permutations of $U$ and $B_c$ reduces a
ballot to $(z,e,b)$, where $z=|A_i\cap U|$, $e=|A_i\cap B_c|$, and
$b=\1[c\in A_i]$.  These ranges give $9\cdot3\cdot2=54$ exact voter orbits.
The strict-gainer indicator and drift column are
\[
                 g=\1[e>z+b],\qquad
                 d=\frac{e/2-b}{z+1}.
\]
In zero-based verifier order, row 0 sums the $16$ one-swap PAV inequalities,
row 5 averages the competitor obtained by removing two members of $U$ and
adding $c$ and one member of $B_c$, and row 7 is $-g\le-2/9$.
The JSON record stores
\[
 (\lambda_0,\lambda_5,\lambda_7)
   =\left(\frac1{144},\frac79,\frac{17}{16}\right),
 \qquad \mu=-\frac29.
\]
The verifier checks $d+\lambda_0R_0+\lambda_5R_5+\lambda_7R_7\ge\mu$ on
all 54 orbits.  Summing against voter masses gives
\[
 \frac{\delta_c}{n}\ge
 -\frac29-\frac{17}{16}\left(-\frac29\right)=\frac1{72}.
\]
For example, orbit $(z,e,b)=(0,1,0)$ has
$(d,R_0,R_5,R_7)=(1/2,8,1/2,-1)$.  Its certified left side is
$-17/144\ge-2/9$.

For a two-blocker example, take the archived motif
\[
                     \texttt{r2\_r2\_m1\_s1\_new0\_p3}
\]
with Venn counts
$(\nu_{00},\nu_{10},\nu_{01},\nu_{11})=(1,0,0,7)$.
Thus $H_c=H_x=H$ is a seven-set,
$U\setminus H=\{u_0\}$, $B_c=\{x,a\}$, and $B_x=\{c,a\}$.
A ballot orbit is $(p,q;b_c,b_x,b_a)$, recording approval of $u_0$, the
number approved in $H$, and the three outsider bits.  There are
$2\cdot8\cdot2^3=128$ such orbits.  Verifier row 4 is the averaged PAV row
for $U-u_0+a$, and rows 9 and 10 are the two scaled blocking rows.
The stored certificate has
\[
 \lambda_4=\frac89,\qquad \lambda_9=\frac8{81},\qquad
 \lambda_{10}=\frac1{81},\qquad \rho=\frac89.
\]
For every orbit the verifier checks
\[
 \rho+\lambda_4R_4+\lambda_9R_9+\lambda_{10}R_{10}\ge0.
\]
The common-scale coefficient is
$9(8/81+1/81)=1$, while the resulting bound is $\beta=8/9<1$.
Hence a genuine pair, which has scale one, cannot realize this motif.
For the ballot approving only $a$, the row values are $(1,-9,-9)$ and the
displayed left side is $7/9\ge0$.  These examples use the same zero-based row
indices and rational multipliers as the released JSON files.

\section{Discussion and conclusion}

Theorem~\ref{thm:main} establishes deterministic Hare-core nonemptiness at
$k=9$.  Together with the previously known $k\le8$ theorem, it yields
nonemptiness for every $k\le9$.  The argument explains why local blocker
escape is sufficient only after all responses are coupled: the stationary
identity uses the same voter profile and the same PAV marginal at every state.
Together with the known candidate and voter-type bounds, the result confines
any unresolved instance to $k\ge10$, at least 16 candidates, and at least
eight distinct approval types.

The portable part of the proof is Lemma~\ref{lem:portable}.  At another seat
number, the same closure argument would apply if one could certify a uniform
one-blocker loss $L$ and a simultaneous successor gain $M>L$.  The constants
and equality analysis established here are specific to the eight-to-nine
extension, so they do not by themselves provide such certificates for larger
committees.

Accordingly, this paper neither proves nonemptiness for arbitrary $k$ nor
supplies an empty-core election.  It also makes no polynomial-time claim for
finding the certified anchor or its stable extension.  Its exact conclusion
is that every finite nine-seat approval election has a deterministic
Hare-quota core committee.

\section*{AI use disclosure}

OpenAI Codex, using GPT-5.6 Sol, was used for proof exploration,
exact-verifier development, literature-query formulation, manuscript
drafting, and release packaging.  AI-generated output is not treated as
mathematical evidence.  The results are supported by the symbolic proof,
the frozen exact rational certificates, and the independently replayable
verifiers described above.  The author reviewed the statements, proofs,
citations, source code, and final artifacts and remains responsible for
the content.

\section*{Code availability}

The verifier source code, exact rational certificates, manuscript source,
and reproduction instructions are available at
\url{https://github.com/Baymax-ray/approval-core-nine-seats}.

\small
\raggedright

\end{document}